\documentclass[11pt]{article}

\usepackage{amsmath,amssymb,amsthm,mathtools}
\usepackage[T1]{fontenc}
\usepackage{lmodern}
\usepackage{geometry}
\usepackage{graphicx}
\usepackage{hyperref}
\title{Consensus and Persistent Harmonic Edge Circulation in a Hodge-Theoretic Model of Networked Information Flow}
\author{Moses Boudourides\\School of Professional Studies\\Northwestern University\\\texttt{Moses.Boudourides@northwestern.edu}}
\date{}

\newtheorem{theorem}{Theorem}
\newtheorem{proposition}{Proposition}
\newtheorem{corollary}{Corollary}
\theoremstyle{definition}
\newtheorem{remark}{Remark}

\newcommand{\R}{\mathbb{R}}
\newcommand{\ip}[2]{\left\langle #1,#2\right\rangle}
\newcommand{\norm}[1]{\left\lVert #1\right\rVert}
\newcommand{\1}{\mathbf{1}}

\begin{document}
\maketitle

\begin{abstract}
We propose a finite-dimensional cochain model for information flow on an online communication complex. Node variables represent issue positions, while edge variables represent independently modelled signed information flow. The coupling is written in terms of the coboundary operator $d_0$ and the $1$-cochain Hodge Laplacian $\Delta_1=d_0d_0^*+d_1^*d_1$. We prove conservation of the mean opinion, a Lyapunov energy law, and convergence to an equilibrium determined by the initial harmonic projection of the edge flow. In particular, node opinions converge to consensus for every initial condition, whereas the edge flow converges to $P_{\mathcal H^1}u_0$. Thus, trivial first cohomology implies decay of the entire edge-flow variable. Nontrivial first cohomology supplies a topological capacity for residual edge circulation. A particular trajectory has a nonzero residual circulation if and only if $P_{\mathcal H^1}u_0\neq 0$, in which case the limiting edge flow is $P_{\mathcal H^1}u_0$. The model therefore establishes that node consensus need not imply decay of an independently represented edge-flow variable; it does not model the formation, reinforcement, or amplification of behavioral echo chambers. We also consider linear damping and a bounded nonlinear saturation as modifications that remove persistent edge flow under the stated assumptions.
\end{abstract}


\section{Introduction}

Classical opinion-dynamics models describe how agents update beliefs by repeated averaging, social influence, or signed interaction \cite{DeGroot1974,FriedkinJohnsen2011,Altafini2012,GrabischEtAl2020}. In parallel, discrete exterior calculus and Hodge theory provide a natural language for distinguishing gradients, circulations, and harmonic components on graphs and cell complexes \cite{GradyPolimeni2010,BattistonEtAl2020,ZieglerEtAl2021,Lim2015}. This suggests a simple mathematical question: can one build a cochain dynamical system in which democratic information flow is decomposed into disagreement-driven transport and topologically persistent circulation?

The present note proposes a minimal phenomenological node--edge model for this question. A $0$-cochain records an aggregate issue position at each node, while a $1$-cochain records an independently modelled signed intensity of information flow along oriented edges. Keeping these variables separate is deliberate: node positions describe the state on which consensus is assessed, whereas the edge cochain records a transport-like field that may retain circulation after those positions have converged. The coupling follows a simple node--edge bookkeeping principle. The divergence $d_0^*u$ changes node positions through the net imbalance of incident edge flow; the coboundary $d_0x$ converts node disagreement into an exact edge-flow contribution; and the Hodge Laplacian dissipates the non-harmonic part of that edge field. The purpose is not to derive these laws from individual message-level behavior, but to isolate the mathematical question of whether node consensus requires decay of a separately represented edge-flow variable. The answer is governed by the initial harmonic projection $P_{\mathcal H^1}u_0$.

The present model lies at the intersection of several established strands of work, but it is not intended to replace them. Combinatorial Hodge theory has long been used to decompose static edge flows into gradient, curl, and harmonic components; HodgeRank is a prominent example in which this decomposition is used to analyze pairwise-comparison data \cite{JiangLimYaoYe2011,Lim2015}. Those works supply the algebraic interpretation of edge flows and harmonic components used here, but they do not study the present coupled node--edge ODE.

There is also direct nearby work on dynamics of higher-order and edge signals. Hansen and Ghrist develop discourse sheaves and sheaf-Laplacian diffusion dynamics for opinion dynamics with richer communication mechanisms, including selective expression and lying \cite{HansenGhristOpinionSheaves}. Ziegler et al.\ study edge consensus on simplicial complexes using balanced Hodge Laplacians and convergence toward a homology-space subspace \cite{ZieglerEtAl2021}. The graph specialization of the present system is directly related to established port-Hamiltonian network dynamics: it admits a dissipative port-Hamiltonian representation and, after an invertible change of edge variable, an equal-parameter mass--spring--damper form on the graph \cite{VanDerSchaftMaschke2013}. Eliminating the edge variable on exact modes also yields the standard damped second-order consensus equation \cite{Ren2008}. Recent topological-Dirac models likewise couple signals on nodes and links, but address wave dynamics, reaction--diffusion instabilities, or synchronization rather than the present consensus problem \cite{Bianconi2021,GiambagliEtAl2022,CalmonEtAl2022}.

Accordingly, the contribution of the present note is deliberately narrow. We do not claim novelty for Hodge decomposition, harmonic invariance, incidence-based node--edge coupling, the basic graph-level convergence mechanism, or the associated second-order consensus behaviour. On graphs, these structural features are already represented within the port-Hamiltonian and mass--spring--damper framework \cite{VanDerSchaftMaschke2013}. The contribution here is instead to formulate the node--edge system explicitly in cochain and Hodge-theoretic terms on cell complexes, to separate its exact, coexact, and harmonic edge sectors, and to identify the limiting edge state as the initial harmonic projection. The cell-complex formulation also distinguishes visible graph cycles from genuine first-cohomology directions after $2$-cells have been taken into account, and it separates the topological capacity for a residual circulation from its excitation by the initial edge state. The stated damping and saturation variants then show how full-edge-flow dissipation removes the neutral harmonic sector under their respective assumptions.

From the viewpoint of nonlinear science, the model offers a minimal setting in which topology, dissipation, and network structure determine long-time node--edge dynamics. Recent work on higher-order networks emphasizes that group structure can qualitatively change synchronization, contagion, and consensus phenomena \cite{BattistonEtAl2020,MajhiPercGhosh2022}. The present formulation isolates a mathematically explicit persistence mechanism at the edge-flow level: harmonic directions are neutral under the Hodge dissipation and, when initially excited, remain after node opinions have converged. This is a statement about the persistence of an independently represented cochain variable, not a claim that the model forms an echo chamber.

Broader network opinion-dynamics research studies mechanisms that are deliberately held fixed in the present note. Higher-order opinion models specify group interactions directly, nonlinear multioption models can generate multistable agreement and disagreement through state-dependent influence, and coevolving-network models couple opinion updates to rewiring or other changes in the interaction structure \cite{ZhangXuZhangChen2024,BizyaevaFranciLeonard2023,Min2023}. Such models can address polarization, fragmentation, adaptive sensitivity, or changing social ties. By contrast, the present fixed-complex cochain ODE neither updates the communication topology nor specifies a behavioral influence rule; it isolates the consequence of Hodge dissipation for a separately represented edge-flow variable.

\section{Cochain setting and model}

Let $K=(V,E,F)$ be a finite oriented $2$-dimensional cell complex whose $1$-skeleton is connected. We write $C^0(K)$, $C^1(K)$, and $C^2(K)$ for the spaces of real-valued $0$-, $1$-, and $2$-cochains. Throughout the main analysis and all baseline simulations, we use the unweighted Euclidean inner products
\[
\ip{x}{y}_0:=\sum_{v\in V}x(v)y(v),\qquad
\ip{u}{w}_1:=\sum_{e\in E}u(e)w(e),\qquad
\ip{p}{q}_2:=\sum_{f\in F}p(f)q(f),
\]
and write $\norm{\cdot}$ for the associated norm on the relevant cochain space. Thus the adjoints in this article are Euclidean adjoints. In matrix notation, for a weighted extension with positive diagonal weight matrices $W_0$, $W_1$, and $W_2$, the corresponding adjoints would be
\[
d_0^*=W_0^{-1}d_0^{\mathsf T}W_1,
\qquad
d_1^*=W_1^{-1}d_1^{\mathsf T}W_2.
\]
The present paper fixes $W_0=I$, $W_1=I$, and $W_2=I$. Consequently, the harmonic projector $P_{\mathcal H^1}$ is the Euclidean orthogonal projector. Under nontrivial weights, both the adjoints and this projector would change; the capacity statement $\dim\mathcal H^1$ remains topological, whereas the realized projection $P_{\mathcal H^1}u_0$ is weight-dependent.

The coboundary operator is
\[
d_0:C^0(K)\to C^1(K),
\]
and its adjoint is denoted by $d_0^*$. If $d_1:C^1(K)\to C^2(K)$ is the next coboundary, then the $1$-cochain Hodge Laplacian is
\[
\Delta_1:=d_0d_0^*+d_1^*d_1.
\]
The space of harmonic $1$-cochains is
\[
\mathcal H^1:=\ker \Delta_1=\ker d_0^*\cap \ker d_1.
\]
When $K$ is a connected graph with no $2$-cells, the harmonic $1$-cochain space is naturally identified with the graph cycle space, whose dimension is
\[
\dim \mathcal H^1=|E|-|V|+1,
\]
the cycle rank of the connected graph \cite{GradyPolimeni2010,Lim2015}. Individual graph cycles need not be linearly independent, but a cycle-space basis provides this number of independent harmonic directions. In particular, trees satisfy $\mathcal H^1=\{0\}$, whereas an unfilled graph cycle may support a nontrivial harmonic mode.

For a genuine $2$-complex, however, a graph cycle and a first-cohomology class are not the same object. A graph cycle is a closed $1$-chain in the $1$-skeleton. A first-cohomology class is represented here by a harmonic $1$-cochain, equivalently by a cochain in $\ker d_0^*\cap\ker d_1$. If a graph cycle bounds a $2$-cell, it is a filled cycle: its circulation is measured by $d_1u$ and it does not produce an independent harmonic direction. Thus filling a cycle can reduce $\dim\mathcal H^1$ even though the same loop remains visible in the $1$-skeleton. Conversely, an unfilled cycle can contribute a first-cohomology direction. This distinction is essential below: $\dim\mathcal H^1$ measures the number of independent harmonic directions available after the $2$-cells have been taken into account, not merely the number of visible graph loops.

Accordingly, nontrivial harmonic space is the algebraic signature of a topological capacity for residual edge circulation. It does not, by itself, imply that a particular trajectory has nonzero residual flow: that requires $P_{\mathcal H^1}u_0\neq 0$.

The dynamical variables are a node-opinion field $x(t)\in C^0(K)$ and an edge-flow field $u(t)\in C^1(K)$. We interpret $x_i(t)$ as the current stance of agent or group $i$ on a given issue, and $u_e(t)$ as signed information flow along oriented edge $e$. Our basic model is
\begin{align}
\dot x &= -d_0^*u, \label{eq:x}\\
\dot u &= -\beta\Delta_1u+d_0x, \qquad \beta>0. \label{eq:u}
\end{align}
All variables and time in this minimal model are nondimensionalized. Thus $x$, $u$, and $t$ are dimensionless model variables, and $\beta>0$ is a dimensionless relative dissipation parameter after the node--edge coupling time scale has been fixed to one. This convention is adopted because the model is phenomenological rather than calibrated to a physical unit of message volume, exposure, or time.

The orientation of each edge is a fixed algebraic reference orientation used to define cochain coordinates and the signs in $d_0$ and $d_1$. For an edge oriented from $i$ to $j$, a positive value $u_e$ denotes flow in that chosen reference direction and a negative value denotes flow in the reverse direction. Reversing the reference orientation changes the sign of the corresponding cochain coordinate and the associated incidence entries, but not the underlying geometric edge flow, the Hodge decomposition, or the conclusions of the analysis. This arbitrary cochain orientation must not be confused with a genuinely directed communication network. The present complex represents pairwise communication relations with signed flow coordinates; it does not specify asymmetric tie existence, sender-specific transmission rules, or directed platform architecture.

In this phenomenological interpretation, the first equation is a balance law: an oriented edge-flow imbalance, represented by $-d_0^*u$, changes the node-position field. The second equation adds an exact flow induced by node disagreement, $d_0x$, while $-\beta\Delta_1u$ dissipates non-harmonic edge flow. Thus the model is appropriate for studying the structural compatibility of node consensus and residual edge circulation; it is not a micro-founded description of how users create, select, interpret, or retransmit individual messages.

\begin{remark}[Model role and scope]
The ODE system is a minimal phenomenological representation of coupled node positions and edge flow. Its intended role is to determine what the cochain structure and Hodge dissipation alone imply for long-time node--edge dynamics. It is not a micro-founded behavioral model and it does not represent polarization, cluster consensus, homophily, selective exposure, recommender systems, confirmation bias, stubbornness, message semantics, heterogeneous user activity, time-varying network formation, or empirical exposure data. Consequently, the model does not establish behavioral echo-chamber formation or a platform-design intervention. Within the limited topological analogy used here, topology determines the harmonic space $\mathcal H^1$ and hence the capacity for residual circulation, the initial condition determines whether that capacity is excited through $P_{\mathcal H^1}u_0$, and the long-time residual, when present, is exactly $P_{\mathcal H^1}u_0$.
\end{remark}

For $x\in C^0(K)$, let
\[
\bar x:=\frac{1}{|V|}\sum_{v\in V}x(v),
\qquad
x^\perp:=x-\bar x\1.
\]
Because the graph is connected, $\ker d_0=\operatorname{span}\{\1\}$.

\section{Basic invariants and the energy law}

The first property is conservation of the mean opinion.

\begin{proposition}[Mean opinion is conserved]\label{prop:mean}
Every solution of \eqref{eq:x}--\eqref{eq:u} satisfies
\[
\frac{d}{dt}\bar x(t)=0.
\]
Hence the average opinion remains equal to its initial value.
\end{proposition}

\begin{proof}
Since $d_0\1=0$, one has $\ip{\1}{d_0^*u}=\ip{d_0\1}{u}=0$. Therefore
\[
\frac{d}{dt}\sum_{v\in V}x(v)=\ip{\1}{\dot x}=-\ip{\1}{d_0^*u}=0.
\]
Dividing by $|V|$ gives the result.
\end{proof}

The natural Lyapunov function is
\[
E(x,u):=\norm{x^\perp}^2+\norm{u}^2.
\]

\begin{theorem}[Energy identity]\label{thm:energy}
Every solution of \eqref{eq:x}--\eqref{eq:u} satisfies
\begin{equation}
\frac12\frac{d}{dt}E(x(t),u(t)) + \beta\norm{\Delta_1^{1/2}u(t)}^2 =0,
\label{eq:energy}
\end{equation}
where $\norm{\Delta_1^{1/2}u}^2=\ip{\Delta_1u}{u}$, for all $t\geq 0$. In particular, $E$ is nonincreasing. This identity establishes dissipation of the non-harmonic edge-flow component, but by itself it does not specify a convergence rate for the coupled node--edge system.
\end{theorem}

\begin{proof}
Since $\bar x$ is constant by Proposition~\ref{prop:mean}, we have $\dot x^\perp=\dot x=-d_0^*u$. Hence
\[
\frac12\frac{d}{dt}\norm{x^\perp}^2=\ip{x^\perp}{\dot x}=-\ip{x^\perp}{d_0^*u}=-\ip{d_0x^\perp}{u}=-\ip{d_0x}{u},
\]
because $d_0\1=0$. Also,
\[
\frac12\frac{d}{dt}\norm{u}^2=\ip{u}{\dot u}=-\beta\ip{\Delta_1u}{u}+\ip{u}{d_0x}.
\]
The mixed terms cancel, giving
\[
\frac12\frac{d}{dt}E(x,u)=-\beta\ip{\Delta_1u}{u}=-\beta\norm{\Delta_1^{1/2}u}^2.
\]
This is \eqref{eq:energy}.
\end{proof}

\section{Asymptotic behavior and residual harmonic edge flow}

The harmonic component of the information flow is invariant.

\begin{proposition}[Harmonic component is frozen]\label{prop:harmonic}
Let $P_{\mathcal H^1}$ denote the orthogonal projector onto $\mathcal H^1$. Then
\[
P_{\mathcal H^1}u(t)=P_{\mathcal H^1}u(0)
\qquad\text{for all } t\ge 0.
\]
\end{proposition}

\begin{proof}
Apply $P_{\mathcal H^1}$ to \eqref{eq:u}. Since $P_{\mathcal H^1}\Delta_1=0$ and $P_{\mathcal H^1}d_0=0$ by orthogonality of exact and harmonic $1$-cochains, one gets
\[
\frac{d}{dt}P_{\mathcal H^1}u=0.
\]
\end{proof}

\begin{corollary}[Capacity, excitation, and residual flow]\label{cor:capacity-excitation}
Let $P_{\mathcal H^1}$ be the orthogonal projector onto $\mathcal H^1$. The complex has capacity for residual harmonic edge flow if and only if $\mathcal H^1\neq\{0\}$. For a given initial condition, this capacity is excited if and only if $P_{\mathcal H^1}u_0\neq 0$. In that case, the residual edge flow is nonzero and equals $P_{\mathcal H^1}u_0$; if $P_{\mathcal H^1}u_0=0$, then the limiting edge flow vanishes.
\end{corollary}

We now identify the long-time limit.

\begin{theorem}[Consensus with residual harmonic edge flow]\label{thm:main}
Let $K$ be connected. For every initial condition $(x_0,u_0)\in C^0(K)\times C^1(K)$, the solution of \eqref{eq:x}--\eqref{eq:u} converges globally to the equilibrium
\[
(x(t),u(t))\to (\bar x_0\1,P_{\mathcal H^1}u_0)
\qquad \text{as } t\to\infty,
\]
where $\bar x_0$ is the initial average opinion.
\end{theorem}

\begin{proof}
By Theorem~\ref{thm:energy}, the energy is bounded and nonincreasing, so every trajectory is precompact in the finite-dimensional phase space. LaSalle's invariance principle implies that the omega-limit set is contained in the largest invariant subset of
\[
\{(x,u): \norm{\Delta_1^{1/2}u}=0\}=
\{(x,u): u\in \mathcal H^1\}.
\]
If $u\in \mathcal H^1$, then $d_0^*u=0$, so \eqref{eq:x} gives $\dot x=0$. Invariance under \eqref{eq:u} then requires $d_0x=0$, hence $x$ is constant because the $1$-skeleton is connected. Therefore every limit point lies in the equilibrium manifold
\[
\mathcal M:=\{(c\1,h): c\in \R,\ h\in \mathcal H^1\}.
\]
By Proposition~\ref{prop:mean}, the constant must equal $\bar x_0$. By Proposition~\ref{prop:harmonic}, the harmonic component must equal $P_{\mathcal H^1}u_0$. Hence the omega-limit set reduces to the singleton $(\bar x_0\1,P_{\mathcal H^1}u_0)$, which proves the claim.
\end{proof}

\begin{remark}[Coupled spectrum and damping regimes]\label{rem:coupled-spectrum}
The convergence statement in Theorem~\ref{thm:main} is qualitative. Although $\Delta_1$ determines the dissipation term in the energy identity, its smallest positive eigenvalue does not by itself determine every decay rate of the coupled $(x,u)$ dynamics. To see the coupling explicitly, let $\phi$ be a nonconstant eigenvector of the graph Laplacian $L_0:=d_0^*d_0$ with eigenvalue $\mu>0$, and consider the associated exact edge direction $d_0\phi$. On the two-dimensional invariant subspace spanned by $(\phi,0)$ and $(0,d_0\phi)$, the modal amplitudes $(a,b)$ satisfy
\[
\dot a=-\mu b,
\qquad
\dot b=a-\beta\mu b,
\]
and hence
\[
\ddot a+\beta\mu\dot a+\mu a=0.
\]
This exact-sector reduction has the form of a damped graph-wave, or second-order consensus, equation: eliminating the independently represented edge-flow variable yields a graph-Laplacian restoring term and mode-dependent damping. The full baseline system also admits the dissipative port-Hamiltonian representation
\[
\dot z=(J-R)\nabla H(z),
\qquad
z=(x,u),
\qquad
H(x,u)=\frac12\bigl(\norm{x}^2+\norm{u}^2\bigr),
\]
where
\[
J=
\begin{pmatrix}
0 & -d_0^*\\
d_0 & 0
\end{pmatrix},
\qquad
R=
\begin{pmatrix}
0 & 0\\
0 & \beta\Delta_1
\end{pmatrix}.
\]
When there are no $2$-cells, the invertible change of variables
\[
p=x,
\qquad
q=u-\beta d_0x
\]
gives
\[
\dot q=d_0p,
\qquad
\dot p=-d_0^*q-\beta d_0^*d_0p.
\]
Thus the graph specialization has the equal-parameter mass--spring--damper form with damping on every graph edge \cite{VanDerSchaftMaschke2013}. These are structural correspondences; the present phenomenological interpretation does not impose a physical inertia interpretation. 
On the exact sector, the characteristic polynomial is
\[
\lambda^2+\beta\mu\lambda+\mu=0,
\]
so the two modal eigenvalues are
\[
\lambda_\pm(\mu)
=
\frac{-\beta\mu\pm\sqrt{\beta^2\mu^2-4\mu}}{2}.
\]
Hence the exact mode is underdamped when $\beta^2\mu<4$, critically damped when $\beta^2\mu=4$, and overdamped when $\beta^2\mu>4$. In the underdamped regime,
\[
\lambda_\pm(\mu)
=
-\frac{\beta\mu}{2}
\pm i\omega_\mu,
\qquad
\omega_\mu
=
\sqrt{\mu-\frac{\beta^2\mu^2}{4}},
\]
so the modal coordinates oscillate about consensus with exponential envelope $e^{-\beta\mu t/2}$. At critical damping there is a repeated root $-2/\beta$; for the two-dimensional modal system, this is a nonsemisimple repeated eigenvalue. In the overdamped regime, both roots are real and negative.

The remaining Hodge sectors are explicit as well. The eigenvalue $0$ has multiplicity $1+\dim\mathcal H^1$, corresponding respectively to constant node states and harmonic edge flows. On the coexact sector $\operatorname{im}d_1^*$, each positive eigenvalue $\nu$ of $d_1^*d_1$ gives the decay eigenvalue
\[
-\beta\nu.
\]
Thus quantitative decay information is determined by the full Hodge-sector spectrum: the exact-sector roots above, the coexact rates $-\beta\nu$, and the neutral mean and harmonic modes. It does not follow from $\lambda_1^+(\Delta_1)$ alone.
\end{remark}

\begin{corollary}[Trivial first cohomology implies full consensus]\label{cor:consensus}
If $\mathcal H^1=\{0\}$, then every solution satisfies
\[
x(t)\to \bar x_0\1,
\qquad
u(t)\to 0.
\]
Thus the model converges to consensus with vanishing edge flow for every initial condition. More generally, when $\mathcal H^1\neq\{0\}$, vanishing limiting edge flow still occurs whenever $P_{\mathcal H^1}u_0=0$.
\end{corollary}

\begin{remark}
Theorem~\ref{thm:main} separates node-level consensus from edge-level persistence. Node opinions always converge to consensus in this linear model, but the edge flow converges to the residual harmonic component $P_{\mathcal H^1}u_0$. Hence a nonzero residual circulation occurs if and only if the initial edge flow has a nonzero harmonic projection. Nontrivial first cohomology provides the capacity for such a residual component, but does not by itself imply that one is excited. This result concerns an independently represented edge-flow variable. It should not be interpreted as establishing the formation, reinforcement, or amplification of a behavioral echo chamber.
\end{remark}

\section{Numerical study}\label{sec:numerics}

We present a reproducible numerical study of the linear model and its two dissipative variants through four controlled comparisons. First, an unfilled $12$-cycle provides a larger connected graph with one harmonic direction and is used to compare zero and nonzero harmonic initialization. Second, two cell complexes use exactly the same oriented $12$-cycle as their complete $1$-skeleton and exactly the same initial node and edge states. The unfilled complex has no $2$-cell and first Betti number one, whereas the filled complex attaches one polygonal $2$-cell along the complete boundary cycle and has trivial first cohomology. Third, a path and bouquets of one, two, and three cycles provide connected complexes with first Betti numbers $0$, $1$, $2$, and $3$, respectively. Fourth, the variant runs use the unfilled $12$-cycle and compare the baseline system with uniform linear damping and componentwise $\tanh$ saturation. All simulations use the Euclidean cochain conventions fixed in the cochain-setting section above; the complete solver, parameter, initial-condition, and code specifications are given below and in the accompanying reproducibility package.

The numerical outputs are used only to illustrate and verify the analytical conclusions for the stated finite-dimensional models. In particular, a nonzero long-time harmonic norm is interpreted as realized residual edge flow after a deliberately chosen harmonic initial component; it is not interpreted as a behavioral echo-chamber measurement.

\subsection{Harmonic excitation on a larger unfilled cycle}

Figure~\ref{fig:harmonic-initialization} presents a controlled pair of runs on the unfilled $12$-cycle, for which $\dim\mathcal H^1=1$. The node initialization and non-harmonic edge-flow component are identical in the two runs; only the harmonic projection of the initial edge flow differs. In the zero-projection run, the harmonic component is absent and the total edge flow decays toward zero, with only a small non-harmonic transient remaining at the finite reporting time. In the excited run, node disagreement again decays whereas the initialized harmonic component remains at norm $1.5$. The comparison isolates harmonic excitation from topological capacity, as stated in Theorem~\ref{thm:main}.
For the $12$-cycle used here, the positive eigenvalues of $L_0$ are $2-\sqrt{3}$, $1$, $2$, $3$, $2+\sqrt{3}$, and $4$. Since $\beta=0.65$, one has $\beta^2\mu\leq 1.69<4$ for every positive mode. Thus all nonconstant exact components lie in the underdamped regime described in Remark~\ref{rem:coupled-spectrum}; the nonmonotone transient in the plotted disagreement norm is consistent with a superposition of such decaying oscillatory components.

\begin{figure}
\centering
\includegraphics[width=0.96\textwidth]{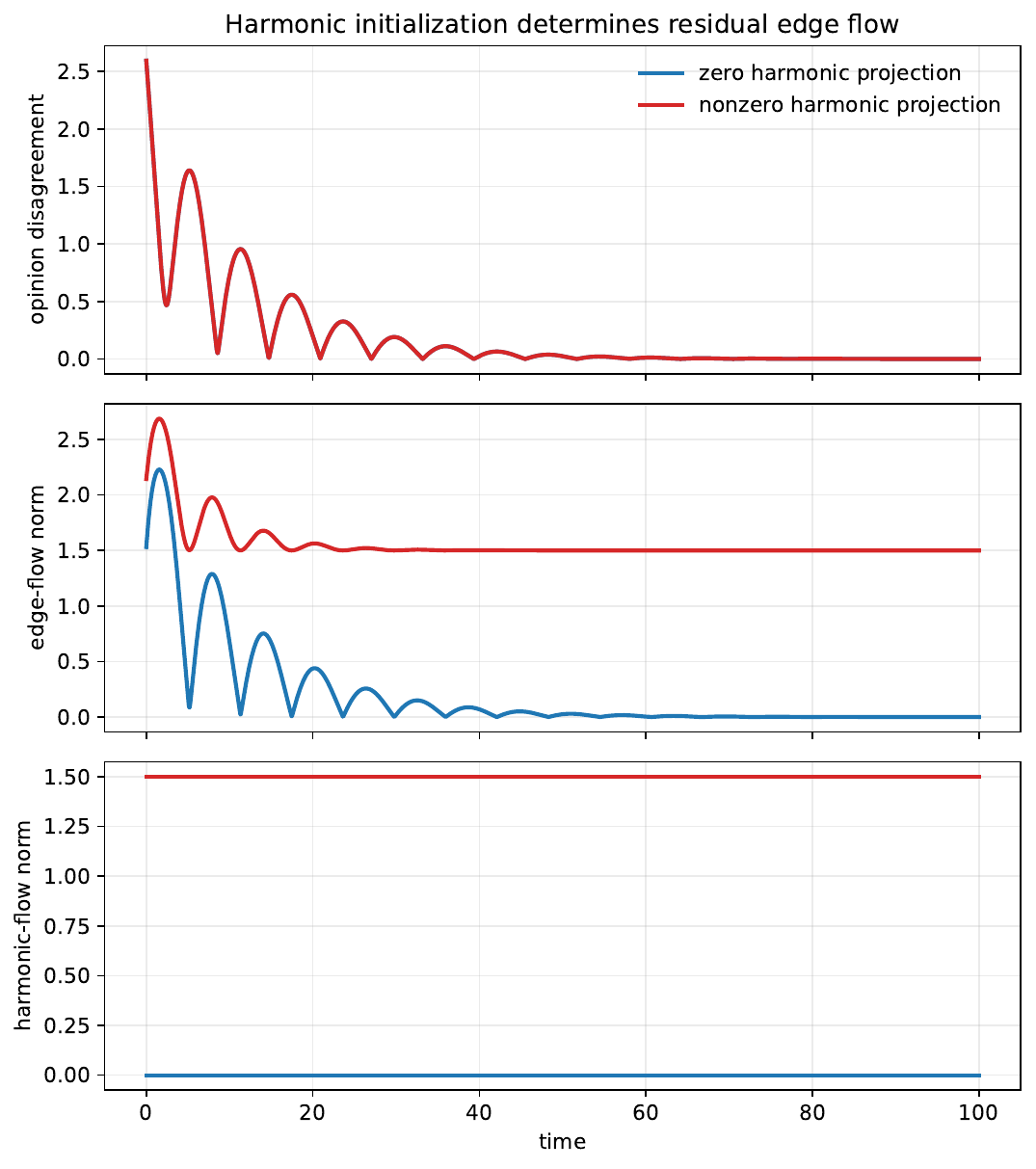}
\caption{Controlled harmonic-initialization comparison on the unfilled $12$-cycle, for which $\dim\mathcal H^1=1$. The two runs differ only in their initial harmonic projection. The zero-projection run has no residual harmonic flow and its total edge flow decays toward zero; the deliberately excited run retains the initialized harmonic component of norm $1.5$. Node disagreement decays in both runs through a nonmonotone transient consistent with the underdamped exact-sector regime described in Remark~\ref{rem:coupled-spectrum}.}
\label{fig:harmonic-initialization}
\end{figure}

\subsection{Filled and unfilled boundary loops}

Figure~\ref{fig:filled-unfilled} compares two cell complexes having exactly the same oriented $12$-cycle as their complete $1$-skeleton. Both simulations use the same node initialization and the same initial edge state, consisting of an exact component together with a boundary circulation of norm $1.5$. The two cases differ only in the presence of one polygonal $2$-cell attached along the full boundary cycle. In the unfilled complex, the boundary circulation is harmonic and is retained. In the filled complex, that identical circulation is coexact, has zero harmonic projection, and is dissipated by the $d_1^*d_1$ term. Indeed, in the filled complex $B_2=\sqrt{12}\,c$, so $d_1^*d_1c=12c$ and the common boundary-circulation amplitude satisfies $\ip{c}{u(t)}=1.5e^{-12\beta t}$; this is the dashed prediction in Figure~\ref{fig:filled-unfilled}. Thus the comparison isolates the effect of the face alone: attaching it removes the harmonic capacity associated with the boundary circulation without changing the $1$-skeleton or the initial state.

\begin{figure}
\centering
\includegraphics[width=0.96\textwidth]{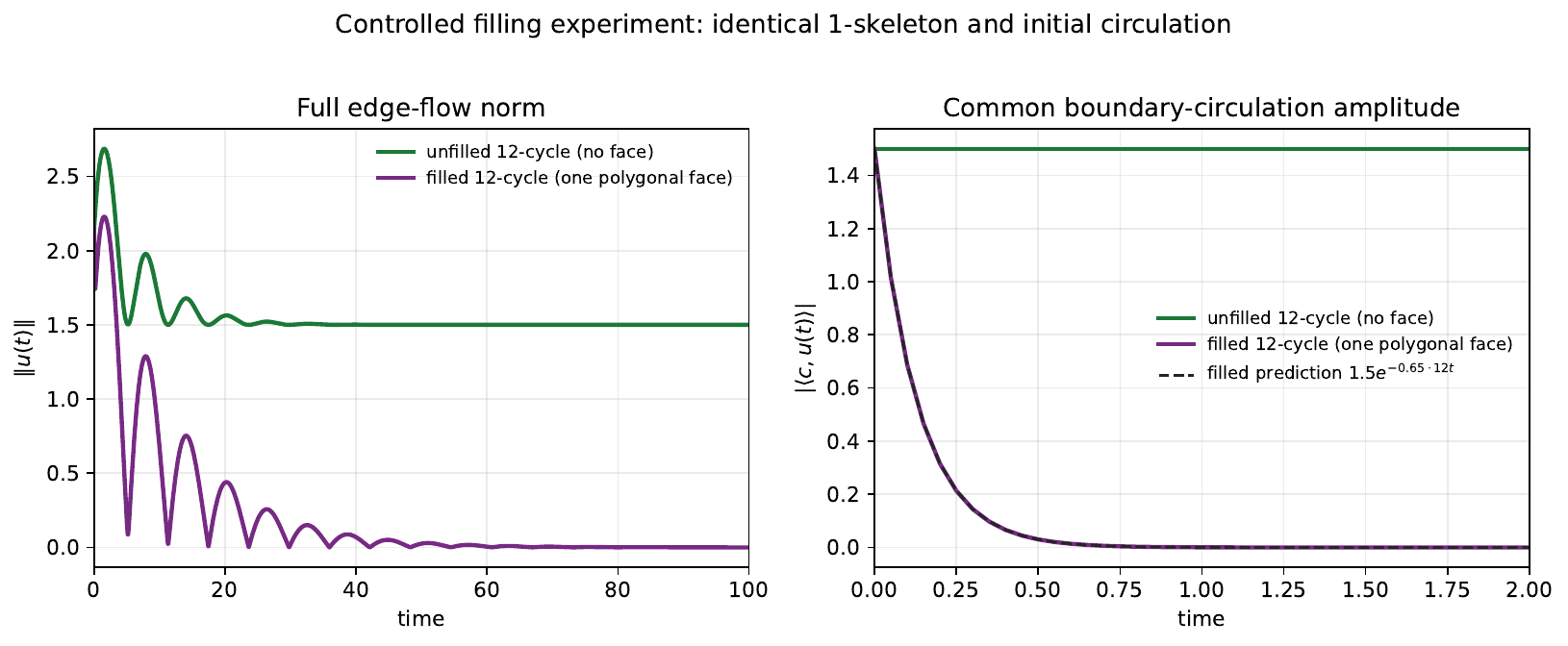}
\caption{Controlled filled-versus-unfilled comparison on the same oriented $12$-cycle and from the same node and edge initial states. The two simulations differ only by the attachment of one polygonal $2$-cell along the complete boundary. Without the face, the initialized boundary circulation is harmonic and persists. With the face, the identical circulation becomes coexact, has zero harmonic projection, and decays under the $d_1^*d_1$ contribution to the Hodge Laplacian.}
\label{fig:filled-unfilled}
\end{figure}

\subsection{Varying first-cohomology rank}

Figure~\ref{fig:cycle-rank} compares a rank-zero path with bouquets carrying one, two, and three independent unfilled cycles. The rank-zero path has no harmonic sector and its edge flow decays. The three bouquet complexes have harmonic spaces of dimensions one, two, and three, respectively. Each bouquet is initialized with a unit-norm harmonic component, which remains while the node disagreement decays. Because the plotted scalar harmonic norms have the same initial value in all three nontrivial cases, the coincident curves demonstrate preservation of equally normalized harmonic components rather than a visual display of the different harmonic-space dimensions. Thus the runs do not show that larger first Betti number creates a residual by itself.

\begin{figure}
\centering
\includegraphics[width=0.96\textwidth]{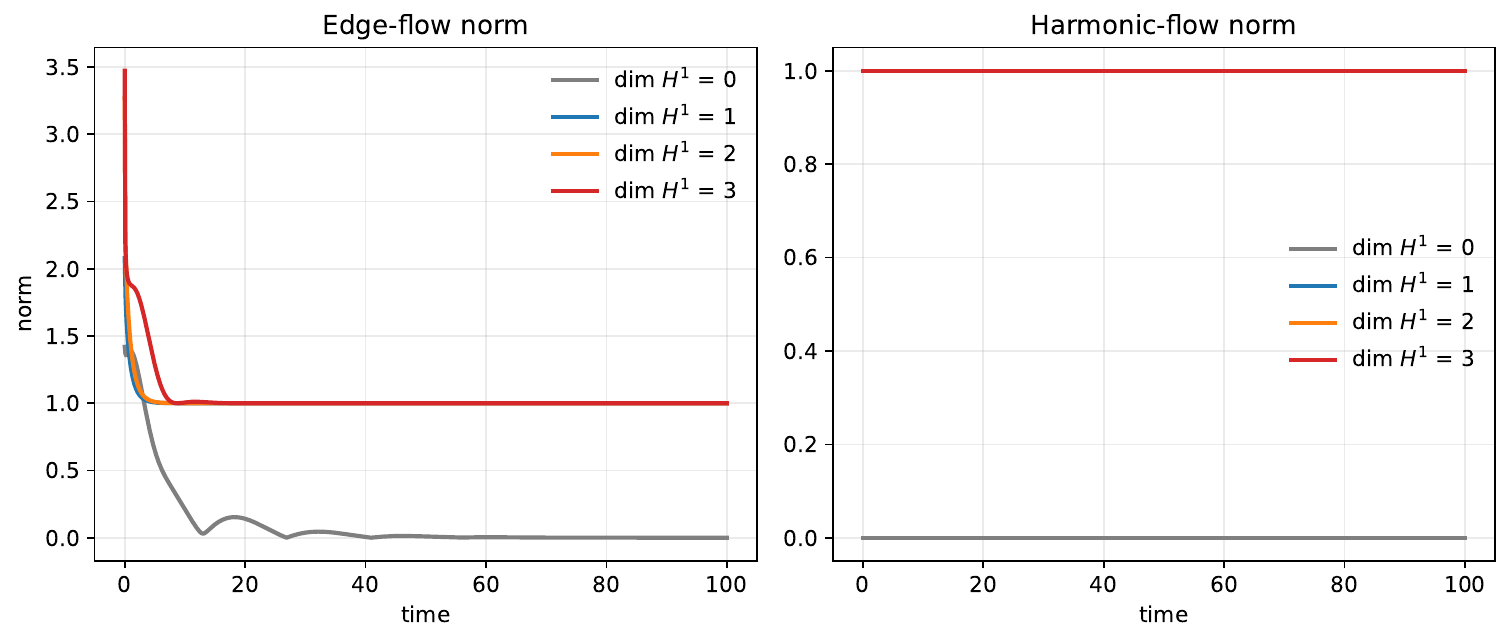}
\caption{Comparison across first-Betti-number ranks $0$, $1$, $2$, and $3$. The rank-zero path has no harmonic sector and its edge flow decays. The nontrivial bouquet complexes have harmonic-space dimensions one, two, and three, respectively, and each is initialized with a unit-norm harmonic component. The coincident harmonic-norm curves demonstrate preservation of equally normalized initialized components, not a visual display of the different harmonic-space dimensions. The comparison does not show that increasing rank automatically generates persistence.}
\label{fig:cycle-rank}
\end{figure}

\subsection{Baseline, damping, and saturation}

Figure~\ref{fig:variants} compares the baseline system on the unfilled $12$-cycle with the linearly damped system and the componentwise $\tanh$-saturation system. The baseline run retains its initialized harmonic component. The damping term removes the full edge-flow variable, and the saturation term likewise removes all residual edge flow under the strict sign condition stated in Proposition~\ref{prop:nonlinear}. These numerical runs illustrate the nonselective effects proved for the two variants in the following analytical section; they do not claim selective suppression of cycle-supported flow.

\begin{figure}
\centering
\includegraphics[width=0.96\textwidth]{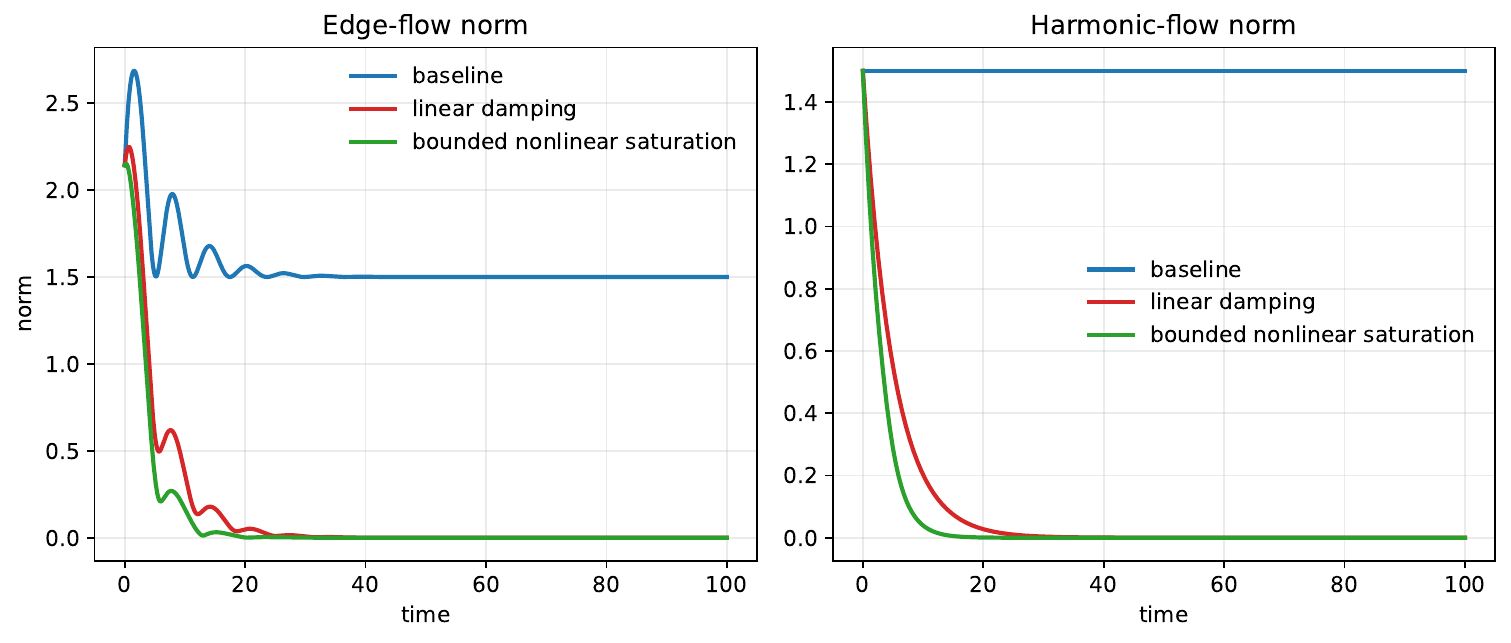}
\caption{Baseline and modified edge-flow dynamics on the same unfilled $12$-cycle with the same initial harmonic norm $1.5$. The baseline retains the initialized harmonic residual. Uniform linear damping and componentwise bounded nonlinear saturation both remove the full edge-flow variable, including its harmonic component, consistent with Propositions~\ref{prop:damped} and~\ref{prop:nonlinear}.}
\label{fig:variants}
\end{figure}

\subsection{Controlled comparisons and terminal diagnostics}

All comparisons use the same deterministic node-initialization formula, the same Euclidean cochain conventions, and a common integration horizon $T=100$; literal equality of node states applies to each controlled pair on a fixed complex. The controlled harmonic-initialization pair on the unfilled $12$-cycle differs only by the amplitude of the initial harmonic projection: its norm is numerically zero in the first run and $1.5$ in the second. At $T=100$, the zero-projection run has harmonic norm below $2.5\times 10^{-16}$ and total edge-flow norm $4.10\times 10^{-4}$, representing a small remaining non-harmonic transient. In contrast, the excited run has terminal total and harmonic edge-flow norms both equal to $1.5$ to the displayed precision. Node disagreement is below $2\times10^{-4}$ in both runs. Thus the comparison changes the initial harmonic component while holding the topology and the non-harmonic initialization fixed, and it confirms the distinction between capacity and realized persistence.

The controlled filled-versus-unfilled comparison uses identical node--edge incidence matrices, identical node initializations, and identical edge initializations in the two runs. The common oriented $12$-cycle has first Betti number one when unfilled and first Betti number zero after one polygonal $2$-cell is attached along its complete boundary. At $T=100$, the unfilled run retains terminal harmonic norm $1.5$ and terminal total edge-flow norm $1.500000056$. The filled run has terminal harmonic norm $2.71\times10^{-19}$ and terminal total edge-flow norm $4.10\times10^{-4}$. The latter is a small remaining exact-sector transient at the finite reporting time, not a residual harmonic component; the full edge flow converges to zero as $t\to\infty$. The validation checks confirm that the two runs have the same $B_1$, $x_0$, and $u_0$, that $B_1B_2=0$, and that the common boundary circulation is harmonic in the unfilled case but coexact with zero harmonic projection in the filled case.

The rank comparison extends the illustration across larger topological variation. The rank-zero path has terminal edge-flow norm $8.77\times10^{-5}$, whereas the bouquets with first Betti numbers one, two, and three retain their deliberately initialized unit-norm harmonic components. Thus the numerical results show equal-norm persistence on harmonic spaces of different dimensions; the scalar harmonic-norm curves do not themselves display the number of available harmonic directions. They do not imply that larger rank automatically creates residual circulation.

Finally, the variant comparison uses the same unfilled $12$-cycle and the same initial harmonic norm $1.5$ in all three runs. The baseline retains terminal harmonic norm $1.5$. With uniform damping coefficient $\eta=0.20$, the terminal total edge-flow norm is $1.88\times10^{-8}$ and the harmonic norm is $3.09\times10^{-9}$. With componentwise saturation $0.30\tanh(u/0.75)$, the terminal total edge-flow norm is $1.99\times10^{-13}$ and the harmonic norm is $9.16\times10^{-18}$. The numerical outcomes therefore agree with the analytical nonselectivity statements: both modifications remove residual edge flow generally, not only harmonic or cycle-supported components.

\subsection{Numerical specification and reproducibility}

The revised numerical evidence is generated by the deterministic implementation openly archived as version 1.1.0 at \href{https://doi.org/10.5281/zenodo.23236037}{https://doi.org/10.5281/zenodo.23236037} and maintained in the version-controlled repository \href{https://github.com/mboudour/hodge-node-edge-information-flow}{https://github.com/mboudour/hodge-node-edge-information-flow}. The release contains the complete simulation and validation scripts, the dependency specification, raw trajectory files, JSON metadata, a numerical-summary file, and the four PDF figures used above. No random number generator is used.

For every complex, an edge listed as $(i,j)$ is oriented from $i$ to $j$, and the node--edge incidence matrix has entry $-1$ at the tail and $+1$ at the head. The $1$-Hodge Laplacian is formed as $L_1=B_1^{\mathsf T}B_1+B_2B_2^{\mathsf T}$. For the controlled filled-versus-unfilled comparison, both complexes use the same oriented $12$-cycle and therefore the same matrix $B_1$. In the unfilled case, $B_2$ has no columns. In the filled case, $B_2$ has one column, whose entries are $1$ on the compatibly oriented boundary edges, representing one polygonal $2$-cell attached along the complete cycle. Thus $B_1B_2=0$. The rank-zero complex is a $13$-node path; the rank-one, rank-two, and rank-three complexes are bouquets of respectively one, two, and three unfilled $8$-cycles sharing a common node.

The baseline, harmonic-initialization, and controlled filled/unfilled runs use $\beta=0.65$. The rank comparison uses $\beta=3.0$ to display the terminal behavior clearly on the common integration horizon. The initial node state is the deterministic zero-mean vector
\[
x_i(0)=\cos\!\left(\frac{2\pi i}{n}\right)+0.35\sin\!\left(\frac{4\pi i}{n}\right),
\]
with its discrete mean subtracted. For the controlled filled/unfilled comparison, both runs use exactly the same edge initialization
\[
u(0)=B_1^{\mathsf T}x(0)+1.5c,
\qquad
c=\frac{1}{\sqrt{12}}(1,\ldots,1)^{\mathsf T},
\]
where $c$ is the unit circulation around the compatibly oriented $12$-cycle. In the unfilled complex, $c$ is harmonic; in the filled complex, it is coexact. For the harmonic-initialization comparison, the amplitude is $a=0$ or $a=1.5$ in $u(0)=B_1^{\mathsf T}x(0)+ah$. The variant runs use $a=1.5$, and the nontrivial-rank bouquet runs use $a=1$, where $h$ is a canonical unit vector in the harmonic subspace constructed from the orthogonal harmonic projector. This projector-based construction avoids dependence on arbitrary eigenbasis choices when the harmonic subspace has dimension greater than one. 

All full runs integrate to $T=100$ at $2001$ reported time points with \texttt{scipy.integrate.solve\_ivp}, using the \texttt{DOP853} method, relative tolerance $10^{-9}$, absolute tolerance $10^{-11}$, and maximum step size $0.10$. The linear-damping variant uses $\eta=0.20$. The nonlinear variant uses the componentwise term $0.30\tanh(u/0.75)$. Before the full runs, the validation script checks $B_1B_2=0$, the first-Betti-number counts, conservation of the node mean, nonincrease of the energy, and basis-invariance of the canonical harmonic initialization for the rank-two and rank-three bouquet complexes. For the controlled filled/unfilled pair, it additionally verifies that the two runs have identical $B_1$, $x_0$, and $u_0$; that the common boundary circulation is harmonic in the unfilled complex; that it is coexact with zero harmonic projection in the filled complex; and that its filled-case decay agrees with the corresponding $d_1^*d_1$ eigenvalue.

\section{Damped and nonlinear variants}

\subsection{Linear damping}

As a mathematical modification that adds direct coercive dissipation to the entire edge-flow variable, consider the linearly damped system
\begin{align}
\dot x &= -d_0^*u, \label{eq:xd}\\
\dot u &= -\beta\Delta_1u-\eta u+d_0x, \qquad \beta>0,\ \eta>0. \label{eq:ud}
\end{align}

\begin{proposition}[Damping removes residual edge flow]\label{prop:damped}
Every solution of \eqref{eq:xd}--\eqref{eq:ud} satisfies
\[
\frac12\frac{d}{dt}E(x(t),u(t))
+\beta\norm{\Delta_1^{1/2}u(t)}^2
+\eta\norm{u(t)}^2=0.
\]
Consequently,
\[
x(t)\to \bar x_0\1,
\qquad
u(t)\to 0
\qquad\text{as } t\to\infty,
\]
even when $\mathcal H^1\neq\{0\}$.
\end{proposition}

\begin{proof}
The proof of the energy identity is the same as before, except for the extra term $-\eta\norm{u}^2$. The LaSalle set is now determined by $u=0$, and invariance then forces $d_0x=0$, hence $x$ constant. Conservation of the mean fixes the limit constant at $\bar x_0$.
\end{proof}

\subsection{Nonlinear edge-flow saturation}

For a nonlinear modification, consider the saturated variant
\begin{align}
\dot x &= -d_0^*u, \label{eq:xn}\\
\dot u &= -\beta\Delta_1u-\gamma\Sigma(u)+d_0x, \qquad \beta>0,\ \gamma>0, \label{eq:un}
\end{align}
where $\Sigma(u)=(\sigma(u_e))_{e\in E}$. We assume throughout this subsection that $\sigma:\R\to\R$ is locally Lipschitz continuous, odd, and strictly sign-preserving:
\[
z\sigma(z)>0\qquad\text{for every }z\neq 0.
\]
Local Lipschitz continuity gives a unique maximal solution for every initial condition, and the sign condition supplies the additional dissipation used below. The canonical bounded choice $\sigma(z)=\tanh z$ satisfies all of these assumptions. The term $-\gamma\Sigma(u)$ acts componentwise on every nonzero edge-flow coordinate; it is therefore a nonlinear edge-flow dissipation, not a selective harmonic- or cycle-mode filter.

\begin{proposition}[Saturation removes residual edge flow]\label{prop:nonlinear}
Assume that $\sigma$ is locally Lipschitz continuous, odd, and satisfies $z\sigma(z)>0$ for $z\neq 0$. Then, for every initial condition, system \eqref{eq:xn}--\eqref{eq:un} has a unique global solution and satisfies
\[
\frac12\frac{d}{dt}E(x(t),u(t))
+\beta\norm{\Delta_1^{1/2}u(t)}^2
+\gamma\ip{\Sigma(u(t))}{u(t)}=0.
\]
Moreover,
\[
x(t)\to \bar x_0\1,
\qquad
u(t)\to 0
\qquad\text{as }t\to\infty.
\]
\end{proposition}

\begin{proof}
Local Lipschitz continuity of $\sigma$ makes the vector field in \eqref{eq:xn}--\eqref{eq:un} locally Lipschitz, so a unique maximal solution exists. The calculation in Theorem~\ref{thm:energy} gives
\[
\frac12\frac{d}{dt}E(x,u)
=-\beta\norm{\Delta_1^{1/2}u}^2-\gamma\ip{\Sigma(u)}{u}.
\]
Because
\[
\ip{\Sigma(u)}{u}=\sum_{e\in E}\sigma(u_e)u_e\geq 0,
\]
the energy is nonincreasing. The mean opinion is conserved exactly as in Proposition~\ref{prop:mean}. Hence $E$ bounds $u$ and $x^\perp$, while the conserved mean bounds the remaining constant component of $x$. The maximal solution therefore remains bounded in the finite-dimensional state space and is global.

By the strict sign condition, $\ip{\Sigma(u)}{u}=0$ holds if and only if $u=0$. The largest invariant subset of the zero-dissipation set is consequently contained in $\{u=0\}$. On that set, invariance of \eqref{eq:un} requires $d_0x=0$, so connectedness gives $x=c\1$; conservation of the mean fixes $c=\bar x_0$. LaSalle's invariance principle now yields the stated limit.
\end{proof}

The two modifications remove residual edge flow for different mathematical reasons. Linear damping adds the coercive term $\eta\norm{u}^2$ to the energy law and therefore damps every edge-flow component, including exact, coexact, and harmonic components. The nonlinear saturation adds the componentwise dissipation $\gamma\sum_e\sigma(u_e)u_e$ and likewise leaves no nonzero edge-flow component neutrally preserved under the stated strict sign condition. Neither result establishes selective suppression of harmonic modes, cycle-supported flow, or behavioral echo chambers. They are statements about two modified cochain dynamical systems; any intervention interpretation would require a separate behavioural and empirical model.

\subsection{Scope and limited contextual analogy}

The model supplies a minimal mathematical description of coupled node positions and an independently represented edge-flow field. Its harmonic component describes the part of that field that is invariant under the stated linear dynamics. This is not an operational model of an echo chamber: the system has no variables for polarization, community-specific belief clustering, homophily, selective exposure, recommender systems, confirmation bias, or message semantics. In particular, global node consensus in the model is incompatible with treating a nonzero harmonic edge cochain as evidence, by itself, of a behavioral echo chamber.

The cited literature on online curation, exposure, and polarization \cite{KitchensJohnsonGray2020,BandyDiakopoulos2021,KelmEtAl2023} provides contextual motivation for studying recirculation in communication systems. It does not validate an identification of harmonic edge flow with any recognized echo-chamber measure. The limited analogy used here is only that a communication complex can support persistent circulation in an independently modelled edge-flow variable after node opinions have converged. Establishing a behavioral interpretation would require additional state variables, an operational observable, and empirical validation beyond the present model.

\subsection{Topological capacity and realized residual edge flow}

The model distinguishes a topological capacity indicator from an initial-condition-dependent residual-flow indicator. The capacity is
\[
\chi_{\mathrm{top}}:=\dim \mathcal H^1,
\]
which is the first Betti number of the cell complex under the present real-coefficient conventions. It counts independent harmonic $1$-cochain directions after filled cycles have been removed by the $d_1$ term. Thus $\chi_{\mathrm{top}}$ is not a count of all graph cycles in the $1$-skeleton: a loop that bounds a $2$-cell is visible as a graph cycle but contributes no independent first-cohomology class.

The quantity $\chi_{\mathrm{top}}$ gives the dimension of the harmonic subspace available for residual edge flow. It does not measure the amount of residual circulation realized in a particular trajectory. For a given initial condition, define
\[
\rho_\infty(u_0):=\norm{P_{\mathcal H^1}u_0}.
\]
By Theorem~\ref{thm:main}, $\rho_\infty(u_0)$ is exactly the norm of the limiting edge flow. Thus $\chi_{\mathrm{top}}>0$ records capacity, whereas $\rho_\infty(u_0)>0$ records realized persistent harmonic edge circulation. In particular, a complex may have $\chi_{\mathrm{top}}>0$ while still satisfying $\rho_\infty(u_0)=0$ for an initial edge flow with zero harmonic projection. Under the weighted extension described above, the capacity remains topological but the orthogonal projection, and hence $\rho_\infty(u_0)$, depends on the chosen inner product.

\subsection{Mathematical implications of the variants}

The damped and nonlinear variants show mathematically that adding a direct dissipative term can remove residual edge flow. In the present model, uniform linear damping acts on the entire edge-flow variable, and the nonlinear saturation term acts on every nonzero edge-flow component. The analysis therefore does not identify a policy intervention, a selective ranking rule, or a mechanism that preferentially targets behavioral echo chambers. Any platform-design interpretation would require a separate model that represents exposure, ranking, user response, and the relevant empirical outcome measures.

\section{Conclusion}

This paper studied a finite-dimensional cochain model in which node positions and an independently represented edge-flow variable evolve on a connected cell complex. The analysis establishes three linked conclusions. First, the mean node position is conserved and node positions converge to global consensus. Second, the non-harmonic part of the edge flow is dissipated, while the harmonic projection of the initial edge flow is invariant. Consequently, the limiting edge flow is $P_{\mathcal H^1}u_0$: residual harmonic edge circulation occurs precisely when the initial edge flow has a nonzero harmonic projection.

The topological conclusion is therefore conditional rather than generative. Nontrivial first cohomology supplies a capacity for residual harmonic edge flow, but it does not create such a residual without harmonic excitation in the initial condition. The controlled filled-versus-unfilled simulations make this distinction concrete: using the same oriented $1$-skeleton and the same initial boundary circulation, the unfilled complex retains that circulation as a harmonic component, whereas attaching one polygonal $2$-cell makes the identical circulation coexact and removes the corresponding $H^1$ direction. The rank comparisons likewise show preservation of equally normalized harmonic components on complexes whose harmonic spaces have different dimensions; the scalar norm curves do not themselves display those dimensions or imply an automatic increase in persistence with topological rank.

The two modified systems clarify how the baseline neutral harmonic sector can be removed. Uniform linear damping is coercive on the full edge-flow variable, and the componentwise saturation term is dissipative on every nonzero edge-flow component under the stated regularity and strict-sign assumptions. Both variants therefore drive the full edge flow to zero; neither constitutes a selective harmonic-mode filter or a platform intervention. The analytical conclusions are supported by the reproducible numerical comparisons and the openly archived code and data described above.

The model has deliberate limitations. It is a minimal phenomenological node--edge system, not a micro-founded model of message production, exposure, ranking, or user response. It has no variables for polarization, cluster consensus, homophily, selective exposure, recommender systems, confirmation bias, or message semantics. A residual harmonic edge cochain is thus not an operational measure of an echo chamber and does not establish the formation, reinforcement, or amplification of a behavioral echo chamber. The appropriate interpretation is a limited mathematical analogy: node consensus can coexist with persistent circulation in a separately represented edge-flow variable.

Future work can extend this baseline by introducing weighted and directed complexes, time-varying topology, external forcing, heterogeneous node dynamics, and additional state variables that represent exposure or ranking. Such extensions would require their own modeling assumptions, analytical treatment, and empirical validation. They should not be regarded as consequences of the present ODE system.

\section*{Data and code availability}
The deterministic code, raw numerical trajectories, JSON metadata, summary metrics, and figure-generation routines supporting the numerical study are openly archived at \href{https://doi.org/10.5281/zenodo.23236037}{https://doi.org/\allowbreak10.5281/\allowbreak zenodo.\allowbreak23236037}. The version-controlled source repository is available at \href{https://github.com/mboudour/hodge-node-edge-information-flow}{https://github.com/\allowbreak mboudour/\allowbreak hodge-node-\allowbreak edge-\allowbreak information-\allowbreak flow}.

\section*{Declaration of generative AI and AI-assisted technologies in the manuscript preparation process}
During the preparation of this work, the author used OpenAI ChatGPT solely for English-language editing and as a native-speaker-style language checker. The tool was not used to generate scientific content, data, analyses, results, or conclusions. The author reviewed and edited all language suggestions as needed and takes full responsibility for the content of the published article.

\bibliographystyle{plainurl}
\bibliography{democracy_gds_rev}

\end{document}